\documentclass[letterpaper, 10 pt, conference]{ieeeconf}  % Comment this line out if you need a4paper

\IEEEoverridecommandlockouts                              % This command is only needed if 
\usepackage{amsmath,amssymb,amsfonts}
\usepackage{soul}
\usepackage{todonotes}
\usepackage{comment}
\usepackage{mathtools}
\usepackage{caption}
\usepackage{subcaption}
\usepackage{hyperref}

\newtheorem{lem}{Lemma}
\newtheorem{prop}{Proposition}
\newtheorem{defin}{Definition}
\newtheorem{ass}{Assumption}
\newtheorem{thm}{Theorem}

\newenvironment{proposition}{\begin{prop}}{\hfill \mbox{\footnotesize$\square$} \end{prop}}
\newenvironment{definition}{\vskip 3pt\begin{defin}}{\vskip 3pt%% \hfill $\square$ 
\end{defin}}

\title{\LARGE \bf
Safe Formation Control of Open Multi-Robot Systems with Connectivity-Preserving Reconfiguration
}

\author{Pelin \c{S}ekercio\u{g}lu, Nicola De Carli, Dimos V. Dimarogonas % <-this % stops a space
\thanks{*This work was supported in part by the Wallenberg AI, Autonomous Systems and Software Program (WASP) funded by the Knut and Alice Wallenberg (KAW) Foundation, the ERC LEAFHOUND Project, the Swedish Research Council (VR), and Digital Futures.}% <-this % stops a space
\thanks{P. \c{S}ekercio\u{g}lu, N. De Carli, and D. V. Dimarogonas are with the Department of Decision and Control Systems, KTH Royal Institute of Technology, SE-100 44 Stockholm, Sweden (e-mail: \{pelinse, ndc, dimos\}@kth.se).}}%

\begin{document}

\maketitle
\thispagestyle{empty}
\pagestyle{empty}

%%%%%%%%%%%%%%%%%%%%%%%%%%%%%%%%%%%%%%%%%%%%%%%%%%%%%%%%%%%%%%%%%%%%%%%%%%%%%%%%
\begin{abstract}
We address the formation control problem for open multi-robot systems (OMRS), i.e., systems in which robots may join or leave the team during operation and new interaction links are established over time, subject to inter-robot collision-avoidance and connectivity-maintenance constraints. The robots are modeled by double integrators, and interact over a dynamic undirected graph.
We design a distributed controller based on the gradient of a barrier-Lyapunov function. % that drives the robots toward the desired formation while avoiding collisions and preserving established interactions. 
To enable team reconfiguration, we introduce a formation manager that coordinates robot additions and removals and establishes prospective edges whenever needed, either to connect a joining robot to the team or to preserve connectivity before a robot departs. The upper-distance constraints of prospective edges are temporarily relaxed through auxiliary dynamics that preserve feasibility while progressively recovering the nominal interaction range.
The resulting open-team dynamics are modeled as a switched system, for which we establish uniform practical stability for almost all initial conditions under a transition-dependent average dwell-time condition. Finally, the proposed approach is validated in realistic Gazebo simulations with dynamically simulated quadrotors undergoing repeated joining and departure maneuvers. \href{https://github.com/KTH-DHSG/constrained_omrs}{[Code]}\href{https://www.youtube.com/watch?v=Hw5_Zn0FGao}{[Video]}

\end{abstract}

%%%%%%%%%%%%%%%%%%%%%%%%%%%%%%%%%%%%%%%%%%%%%%%%%%%%%%%%%%%%%%%%%%%%%%%%%%%%%%%%
\section{INTRODUCTION}
Open multi-robot systems (OMRS) are networks in which robots may join or leave the system over time, causing interactions to be established or lost. Such systems arise naturally in robotic missions. For instance, in aerial inspection of large infrastructures, the limited battery endurance of quadrotors, typically well below the duration of the mission, imposes a rotation scheme in which depleted vehicles return to a charging station while freshly charged ones take their place; unexpected robot failures have the same effect. In all such cases, both the number of robots and the interaction topology vary during operation. Beyond achieving the desired formation, the control design must then also guarantee that inter-robot collisions are avoided and that the links required for coordination are preserved, despite the limited range of the onboard sensing and communication devices \cite{restrepo2022consensus, dimarogonas2007decentralized}. %In all such cases, both the number of robots and the interaction topology may vary during operation, the latter being itself determined by the inter-robot distances and the limited range of the onboard sensing and communication devices \cite{restrepo2022consensus, dimarogonas2007decentralized}.

%\textcolor{red}{Early results on open networks addressed average consensus via discrete-time gossip algorithms with random agent arrivals and departures \cite{hendrickx2017open}, and max-consensus under similar interactions \cite{abdelrahim2017max}. The proportional dynamic consensus problem was subsequently studied in \cite{franceschelli2018proportional} and \cite{franceschelli2020stability}, introducing novel stability notions for open networks. An alternative framework in \cite{xue2022stability} models open networks as switched systems whose models correspond to dynamics of different dimensions over different graph topologies, and establishes stability under a transition-dependent average-dwell-time condition. Building on this framework, \cite{restrepo2022consensus} addressed constrained consensus control for first-order systems, \cite{restrepo2024distributed} proposed a gradient-based controller preserving biconnectivity under agent removal, and \cite{CDC2025OMAS} developed a stability framework for synchronization over open signed networks with time-varying cooperative and antagonistic interactions.}

%================

Early works on open networks addressed average consensus via discrete-time gossip algorithms with random agent arrivals and departures \cite{hendrickx2017open}, max-consensus \cite{abdelrahim2017max}, and proportional dynamic consensus \cite{franceschelli2020stability}. A complementary framework was introduced in \cite{xue2022stability}, where open networks are modeled as switched systems with mode-dependent dimensions and graph topologies, and stability is established under a transition-dependent average-dwell-time condition. Building on this framework, \cite{restrepo2022consensus} addressed constrained consensus for first-order systems, \cite{restrepo2024distributed} proposed a gradient-based controller preserving biconnectivity under agent removal, and \cite{CDC2025OMAS} studied synchronization over open signed networks with time-varying cooperative and antagonistic interactions.

Most of these works assume unconstrained agent states and a prescribed interaction topology. In multi-robot systems, however, prescribed links must respect sensing and communication range limitations, while collision avoidance must also be enforced. Such requirements are naturally expressed as inter-robot constraints and have been addressed using artificial potential functions for single integrators \cite{dimarogonas_connectedness_2008,cheng2014decentralized}, double integrators \cite{5948318}, and unicycles \cite{panagou2013multi}. A widely adopted approach relies on control barrier functions (CBFs), which enforce forward invariance of a safe set through a differential condition on a function that is positive in the interior of the set and vanishes at its boundary \cite{de2024distributed,10556645}. Barrier-Lyapunov functions (BLFs) \cite{panagou2013multi} instead employ Lyapunov-like functions that remain finite in the interior of the constraint set and diverge as its boundary is approached. While CBF-based methods decouple safety from the nominal objective via online optimization, BLFs embed the barrier directly in the Lyapunov function, allowing constraint satisfaction and convergence to be certified jointly, which represents an important advantage in the switched, dimension-varying setting of an open network. 

%===========================

%Constrained consensus and formation problems have been studied for first-order systems \cite{shang2020resilient,chipade2020multi}, second-order systems \cite{verginis_connectivity_2017,grover2022noncooperative}, underactuated UAVs \cite{ER_TAC_drones-sat}, and for robot manipulators \cite{IEEE-robotmanip}; in all of these, however, the set of agents is fixed. \cite{restrepo2022consensus} solves the constrained consensus problem for first-order systems using BLFs.
%The extension to second-order dynamics is not straightforward: the barrier gradient no longer acts directly on the robots' positions, but enters the closed loop through the velocity-error dynamics. Establishing stability therefore requires a composite Lyapunov function combining the barrier with a quadratic term in the velocity errors, and leads to different sufficient conditions on the average dwell time.%, which must account for the jumps undergone by both the position and the velocity errors at each topology change.

%===============================
Constrained consensus and formation control have been studied for first-order systems \cite{shang2020resilient,chipade2020multi}, second-order systems \cite{verginis_connectivity_2017,grover2022noncooperative}, underactuated UAVs \cite{ER_TAC_drones-sat}, and robot manipulators \cite{IEEE-robotmanip}, but with a fixed set of agents. While \cite{restrepo2022consensus} considers constrained consensus for first-order open systems using BLFs, extending the approach to second-order dynamics is nontrivial because the barrier gradient acts through the velocity-error dynamics rather than directly on position. This motivates a composite Lyapunov construction combining the BLF with a quadratic velocity-error term and leads to different sufficient average-dwell-time conditions.

%==============================

In this paper, we study the formation-control problem for open multi-robot systems composed of robots modeled as double integrators and interconnected over an undirected graph that evolves as robots join or leave the team. We %reformulate the control objective in edge-based formation-error coordinates and 
design a distributed controller based on the gradient of a BLF, which drives the robots toward the desired formation while ensuring inter-robot constraints are satisfied. To enable team reconfiguration, we introduce a formation manager that coordinates robot additions and removals and assigns prospective interaction edges whenever new edges must be established. For each prospective edge, we introduce an auxiliary relaxation state whose dynamics temporarily relax the upper-distance constraint while preserving feasibility and drive the relaxation back to zero as the nominal interaction range is recovered, following a philosophy similar to~\cite{mehdifar2025robust,trakas2023robust, decarli2026underwaterformationcontrol}. Joining robots are connected to the existing team through prospective edges whose connectivity constraint is temporarily relaxed until the nominal interaction range is recovered. When the departure of a robot would disconnect the interaction graph, the formation manager instead follows a make-before-break strategy, establishing the required bridging edges before allowing the robot to leave. %In this way, both robot additions and removals can be performed while preserving connectivity of the active team.

%The resulting open-team dynamics are modeled as a switched system over varying interaction topologies and dimensions. Using a composite Lyapunov function combining the barrier and velocity-tracking terms, we establish uniform practical stability of the closed-loop system for almost all initial conditions under a transition-dependent average dwell-time condition, together with preservation of the inter-robot constraints. Finally, the proposed approach is validated in realistic Gazebo simulations with multiple UAVs repeatedly joining the team from, and returning to, a charging station.

The resulting open-team dynamics are modeled as a switched system with varying topology and dimension, for which uniform practical stability and preservation of the inter-robot constraints are established under a transition-dependent average dwell-time condition. The approach is validated in realistic Gazebo simulations with multiple UAVs repeatedly joining the team from, and returning to, a charging station.

%\nicksay{If we manage to shorten a bit the intro it is good.}

\textit{Notation:} $\lvert \cdot \rvert$ denotes the absolute value for scalars, the Euclidean norm for vectors, and the spectral norm for matrices. $\mbox{card}(\cdot)$ indicates the cardinality of a set. %$\mbox{diag}(z)$ denotes a diagonal matrix whose diagonal elements are the entries of the vector $z$. 
$\mathbb{R}$ is the set of real numbers and $\mathbb{R}_{\geq 0}$ the nonnegative orthant. $A>0$ ($A\ge 0$) indicates that $A$ is a positive definite (positive semidefinite) matrix.

\section{Model and Problem Formulation}
\label{sec:model_problem}

We consider an aerial multi-robot team whose composition may change during a long-duration mission as robots join or leave, e.g., to recharge or due to temporary unavailability. The resulting time-varying team is modeled as a switched system \cite{xue2022stability}.

Specifically, let
\(
\sigma:\mathbb{R}_{\geq 0}\rightarrow\mathcal{P}
\)
denote a switching signal describing the current composition of the
team, where $\phi=\sigma(t)\in\mathcal{P}$ denotes the active mode.
For any interval
\(
t\in[t_l,t_{l+1}),
\)
the composition of the team is constant and consists of $N_\phi$
robots, where $t_l$ and $t_{l+1}$ are consecutive switching instants. A switching instant corresponds to a change in the
team composition and/or established interaction topology, such as robot
admission, bridge activation, or robot departure. %A switching instant $t_l$ corresponds to a robot joining or leaving the team \cite{xue2022stability}.

At each mode $\phi\in\mathcal{P}$, the translational dynamics of robot $i \in \{1, 2, \dots, N_{\phi} \}$ are modeled as
\begin{subequations}
\label{SO}
\begin{align}
    \dot x_i &= v_i,
    \label{SO_a}\\
    \dot v_i &= u_i,
    \label{SO_b}
\end{align}
\end{subequations}
where
\(
x_i,v_i,u_i\in\mathbb{R}^3
\)
denote, respectively, $i$th robot's position, velocity, and commanded translational acceleration.

The double-integrator model in \eqref{SO} provides a suitable approximation of the translational motion of a multirotor vehicle when operating away from aggressive maneuvers and under a sufficiently faster low-level attitude and thrust controller~\cite{lee2010geometric}. In this work, we therefore focus on the coordination layer and assume that the commanded acceleration $u_i$ can be realized by the low-level controller.

\subsection{Interaction topology and inter-robot constraints}

The robots exchange information over an undirected interaction graph
\(
    \mathcal{G}_\phi
    =
    (\mathcal{V}_\phi,\mathcal{E}_\phi),
\)
where
\(
\mathcal{V}_\phi
=
\{1,\ldots,N_\phi\}
\)
is the set of robots participating in the mission and
\(
\mathcal{E}_\phi \subseteq \mathcal{V}_\phi^2
\)
is the set of $M_\phi$ active interaction links. We refer to the interaction links as \emph{edges}, denoted $\varepsilon_k$ with 
$k \leq M_\phi$, and to a pair of robots sharing an edge as \emph{neighbors}. Accordingly, $\mathcal{N}_{i,\phi} := \{ j \leq N_\phi : \varepsilon_k = (i,j) \in \mathcal{E}_\phi \}$ denotes the neighbor set of robot $i$.

The interaction topology is prescribed and remains constant between consecutive switching instants. A topology switch may correspond to the activation or removal of interaction edges, to a change in the team composition, or to both. The resulting graph is determined by the reconfiguration mechanism of Section~\ref{sec:open_team}.

Although the interaction topology is prescribed, the corresponding edges can only be physically maintained while neighboring robots remain sufficiently close. Moreover, all robots must preserve a minimum separation distance to avoid collisions. For every edge
\(
\varepsilon_k=(i,j)\in\mathcal{E}_\phi,
\)
define the relative displacement $\delta_k$ and relative distance $d_k$ between two robots as
\begin{equation}
    \delta_k
    :=
    x_i-x_j  
    \qquad 
    d_k \coloneqq \lvert\delta_k\rvert = \lvert x_i-x_j \rvert
    .
    \label{eq:relative_displacement}
\end{equation}
We assume, for simplicity, common lower and upper admissible
inter-robot distances
\(
d_{\min}>0
\)
and
\(
d_{\max}>d_{\min}.
\)
The admissible set associated with mode $\phi$ is then
\begin{equation}
\label{constraints}
    \mathcal{I}_{k,\phi}
    :=
    \left\{
        \delta_k\in\mathbb{R}^3
        :
        d_{\min}
        <
        d_k
        <
        d_{\max}
    \right\},
    \qquad
    \varepsilon_k\in\mathcal E_\phi .
\end{equation}

The lower bound $d_{\min}$ encodes the minimum safety distance to maintain for collision avoidance, whereas the upper bound $d_{\max}$ ensures that every prescribed edge remains within the communication/sensing range of the robots.

Let $\phi_0=\sigma(t_0)$ denote the initial mode. We assume that the
initial interaction graph is connected and that all its active edges
satisfy their corresponding inter-robot constraints.

\begin{ass}\label{ass1}
The initial graph $\mathcal G_{\phi_0}$ is connected and
\(
    \delta_k(t_0)\in\mathcal I_{k,\phi_0},
\) 
\(
    \forall\,\varepsilon_k\in\mathcal E_{\phi_0}.
\)
\end{ass}

By Assumption~\ref{ass1}, $\mathcal G_{\phi_0}$ contains a spanning tree \cite[Section~3.3]{mesbahi}. For subsequent topology changes, connectivity of the interaction graph and admissibility of newly activated edges are enforced by the reconfiguration mechanism introduced in Section~\ref{sec:open_team}.

\subsection{Problem formulation}\label{sec:prob_formulation}

For every mode $\phi$, let the desired formation be characterized by the desired relative displacement associated with each edge. In particular, for
\(
\varepsilon_k=(i,j)\in\mathcal{E}_\phi,
\)
let
\begin{equation}
    \delta_{d_{k,\phi}}
    :=
    x_{i,d,\phi}-x_{j,d,\phi}
    \label{eq:desired_relative_displacement}
\end{equation}
denote the desired relative displacement between robots $i$ and $j$, where $x_{i,d,\phi},x_{j,d,\phi}\in\mathbb{R}^3$ denote their desired positions. %, defined up to a common translation of the formation.
Since only relative positions are relevant, the desired formation is invariant to a common translation of all robot positions.

For each mode $\phi \in \mathcal{P}$, we define the formation error associated with edge
$\varepsilon_k$ as
\begin{equation}
\label{eq:edge_position_error}
    e_{x_{k,\phi}}
    :=
    \delta_k-\delta_{d_{k,\phi}}
    =
    (x_i-x_j)-(x_{i,d,\phi}-x_{j,d,\phi}).
\end{equation}

The objective is to design a distributed control law $u_i$ such that, for each active mode $\phi \in \mathcal{P}$ of the open multi-robot system,

\begin{enumerate}
    \item the prescribed formation is asymptotically achieved,
    \begin{equation}\label{obj1}
        e_{x_{k,{\sigma(t)}}}(t)\rightarrow0,
        \qquad
        %e_{v,k}(t)\rightarrow0,
        v_{i,{\sigma(t)}}(t)\rightarrow0,
        \qquad
        t\rightarrow\infty,
    \end{equation}
    for all active edges $k\leq M_\phi$ and all robots $i\leq N_\phi$;

    \item the inter-robot constraints remain satisfied,
    \begin{equation}\label{obj2}
        d_{\min}
        <
        d_k
        <
        d_{\max},
        \,
        \forall\,\varepsilon_k=(i,j)
        \in\mathcal{E}_{\sigma(t)},
    \end{equation}
    for all $t\geq0$.
\end{enumerate}

%Hence, the control objective combines formation stabilization, collision avoidance, and preservation of the prescribed edges, despite changes in the composition of the robot team.

%kfor all active edges;
%Sorry 

%\NDcol{
\section{Control Design}
\label{sec:control_design}

\subsection{Constrained formation control}\label{sec:constrained_control}

To address the formation-control problem under inter-robot collision avoidance and connectivity-maintenance constraints, we design a distributed controller based on backstepping and the gradient of a \textit{barrier Lyapunov function} (BLF) \cite{mesbahi,panagou2013multi}. We first recall the definition of a BLF.

\begin{definition}\label{def:2}
Consider the system $\dot{x}=f(x)$ and let $\mathcal{I}$ be an open set containing the origin. A BLF is a positive definite function $W : \mathcal{I}\to \mathbb{R}_{\geq 0},\ x \mapsto  W(x)$, that is $\mathcal{C}^1$, satisfies $\dot{W}(x) \leq 0$, and has the property that $W(x)\to \infty,\ \lvert \nabla W(x) \rvert \to \infty$ as $x \to \partial \mathcal{I}$.
\end{definition}

Using \eqref{eq:edge_position_error}, the constraint set \eqref{constraints} associated with edge $\varepsilon_k\in\mathcal E_\phi$ can be expressed in formation-error coordinates as
\begin{align}\label{294}
	\mathcal{I}_{k,\phi} = &\{ e_{x_{k,\phi}} \in \mathbb{R}^3 : \nonumber\\
    &d_{\min} < \lvert e_{x_{k,\phi}} + \delta_{d_{k,\phi}}  \rvert < d_{\max},\ k \le M_{\phi} \}.
\end{align}
For each $k\leq M_\phi$, we associate with $\varepsilon_k$ a BLF \footnote{An explicit instance of the BLF used in the numerical evaluation is provided in Section~\ref{sec:simulation}.} $\tilde W_{k,\phi}:\mathcal I_{k,\phi}\to\mathbb R_{\geq0}$, defined as
\begin{align}\label{eq:blf_edge}
&\tilde{W}_{k,\phi}(e_{x_{k,\phi}}) = \frac{1}{2}  \left[ \lvert e_{x_{k,\phi}} \rvert^2 + B_{k, \phi}(e_{x_{k,\phi}} + \delta_{d_{k,\phi}}) \right],
\end{align}
where $B_{k, \phi}(e_{x_{k,\phi}} + \delta_{d_{k,\phi}})$ is continuously differentiable on $\mathcal I_{k,\phi}$ and encodes the constraints defined in \eqref{obj2}. It is non-negative, satisfies $B_{k, \phi}(\delta_{d_{k,\phi}})=0$ and $B_{k, \phi}(e_{x_{k,\phi}} + \delta_{d_{k,\phi}}) \to \infty$ as either ${\lvert \delta_k \rvert \to d_{\min}}$ or ${\lvert \delta_k \rvert \to d_{\max}}$. Consequently $\delta_{d_{k,\phi}}$ is a global minimum of $B_{k,\phi}$, so that $\nabla B_{k,\phi}(\delta_{d_{k,\phi}})=0$.

The BLF in \eqref{eq:blf_edge} satisfies $\tilde{W}_{\phi,k}(0)=0$, $\nabla\tilde{W}_{\phi,k}(0)=0$, where $\nabla \tilde{W}_{\phi,k}= \left(\frac{\partial \tilde{W}_{\phi,k}}{\partial {e}_{x_{k,\phi}}}\right)^\top$ and $\tilde{W}_{\phi,k}({e}_{x_{k,\phi}})\to \infty$ as either ${\lvert \delta_k \rvert \to d_{\min}}$ or ${\lvert \delta_k \rvert \to d_{\max}}$. Moreover, since $B_{k,\phi}\geq0$, from \eqref{eq:blf_edge}, it follows that there exists $\kappa_1>0$ such that $ \frac{\kappa_1}{2} \left|e_{x_{k,\phi}}\right|^2 \leq \tilde W_{k,\phi}(e_{x_{k,\phi}}).$ An example of the BLF is provided in \eqref{eq:simulation_blf} in Section~\ref{sec:simulation} and depicted on Figure \ref{fig:combined_blf}. Under the choice of $B_{k,\phi}$, we assume that the associated BLF $\tilde W_{k,\phi}$ has exactly two critical points in $\mathcal I_{k,\phi}$: the desired equilibrium $e_{x_{k,\phi}}=0$ and a nondegenerate saddle point  $e_{x_{k,\phi}}^\star\neq0$ induced by the collision-avoidance constraint; see \cite[Appendix I]{ER_TAC_drones-sat}. Moreover, both critical points are isolated. These properties can be verified from the specific form of $B_{k,\phi}$; see \eqref{eq:simulation_blf}. Accordingly, we denote the set of critical points of $\tilde W_{k,\phi}$ by 
\(
    \mathcal W_{k,\phi}
    :=
    \left\{
        0,e_{x_{k,\phi}}^{\star}
    \right\}.
\)
%
%The origin $\{ {e}_{x_{k,\phi}} = 0\}$ is the unique minimum of $\tilde W_{k,\phi}$, while the collision-avoidance term introduces an additional isolated critical point, denoted by $e_{x_{k,\phi}}^{\star}$; see \cite{RESTREPO-THESIS,sekercioglu2024control}. 
% Accordingly, define
% \(
%     \mathcal W_{k,\phi}
%     :=
%     \left\{
%         0,e_{x_{k,\phi}}^{\star}
%     \right\},
% \)
% and
% \(
%     \left|e_{x_{k,\phi}}\right|_{\mathcal W_{k,\phi}}
%     :=
%     \min
%     \left\{
%         \lvert e_{x_{k,\phi}} \rvert,
%         \lvert e_{x_{k,\phi}}-e_{x_{k,\phi}}^{\star} \rvert
%     \right\}.
% \)
%There exists $\kappa_1>0$ such that $ \frac{\kappa_1}{2} \left|e_{x_{k,\phi}}\right|_{\mathcal W_{k,\phi}}^2 \leq \tilde W_{k,\phi}(e_{x_{k,\phi}}).$

% \begin{figure*}[t]
% \centering
% \begin{subfigure}[b]{0.45\linewidth}
%   \centering
%   \includegraphics[width=\linewidth]{Figures/BLF_3D_1.png}
%   \caption{Surface}
%   \label{fig:sub1}
% \end{subfigure}\hfill
% \begin{subfigure}[b]{0.45\linewidth}
%   \centering
%   \includegraphics[width=\linewidth]{Figures/BLF2D.pdf}
%   \caption{Vertical cut}
%   \label{fig:sub2}
% \end{subfigure}
% \caption{Weighted recentered BLF.}
% \label{fig:test}
% \end{figure*}

% \begin{figure}[t]
% \centering
% \subfloat[Surface]{%
%   \includegraphics[width=0.48\linewidth]{Figures/BLF_3D_1.png}%
%   \label{fig:sub1}}
% \hfill
% \subfloat[Vertical cut]{%
%   \includegraphics[width=0.48\linewidth]{Figures/BLF2D.pdf}%
%   \label{fig:sub2}}
% \caption{Weighted recentered BLF.}
% \label{fig:test}
% \end{figure}

\begin{figure}[htbp]
    \centering
    \begin{subfigure}[b]{0.48\linewidth}
        \centering
        \includegraphics[width=\linewidth]{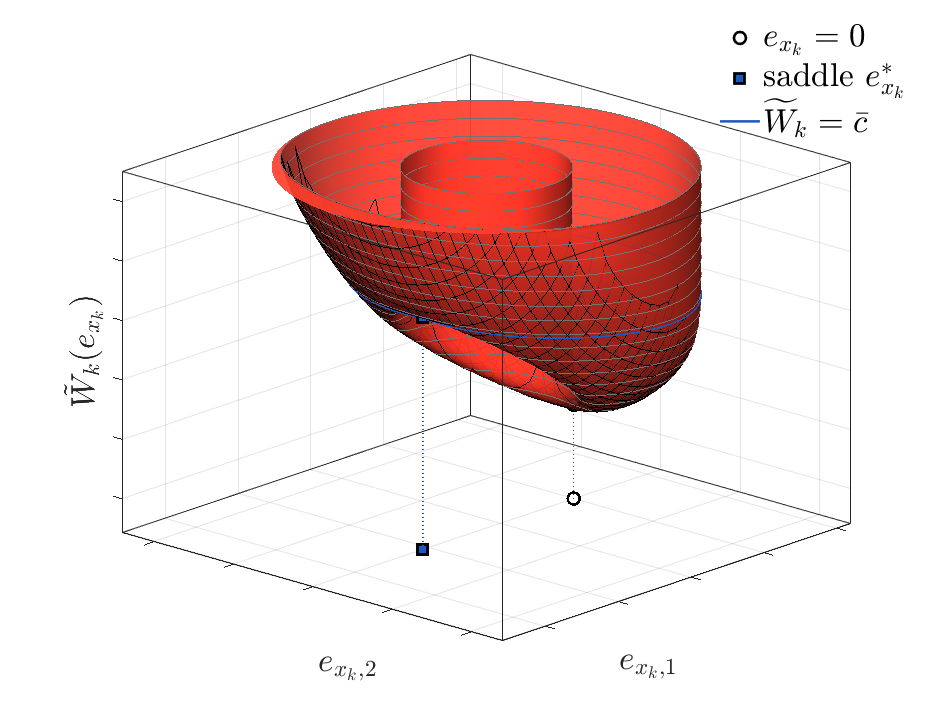}
        %\caption{BLF.}
        \label{fig1}
    \end{subfigure}
    \hfill
    \begin{subfigure}[b]{0.48\linewidth}
        \centering
        \includegraphics[width=\linewidth]{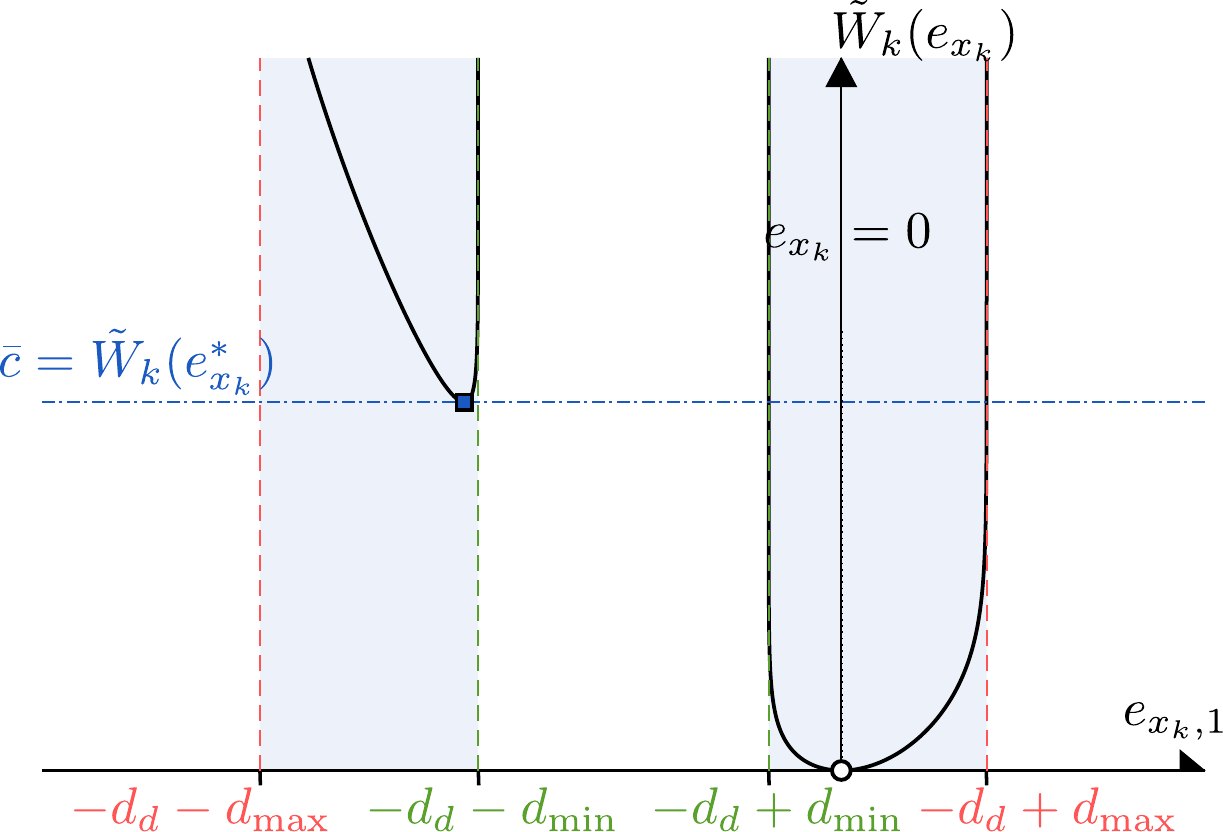}
        %\caption{Vertical cut of the BLF.}
        \label{fig2}
    \end{subfigure}
    \caption{Weighted recentered BLF $\tilde W_{k,\phi}$ in formation-error coordinates. Left: BLF surface, with the desired equilibrium $e_{x_{k,\phi}}=0$, the collision-induced saddle point $e_{x_{k,\phi}}^\star$, and the corresponding critical level $\bar c$. Right: one-dimensional cut illustrating the admissible regions and the divergence of the BLF at the distance constraints.}
    \label{fig:combined_blf}
\end{figure}

% =======

% In addition, we note that $\{ {e}_{x_{k,\phi}} = 0\}$ is a unique minimum of $\tilde{W}_{\phi,k}$ even though a second critical point exists induced by the collision avoidance constraint, which we denote by ${e}^*_{x_{k,\phi}}$ for any $k \leq M_\phi$, see \cite{RESTREPO-THESIS}, \cite[Appendix A.3.2]{sekercioglu2024control}. Define $\mathcal{W}_k := \{ 0, {e}^*_{x_{k,\phi}}\}$ and $\lvert e_{x_{k,\phi}} \rvert_{\mathcal{W}_k} = \min \{ \lvert e_{x_{k,\phi}}\rvert, \lvert e_{x_{k,\phi}} - e^*_{x_{k,\phi}} \rvert \}$ for any $k \leq M_\phi$, then
% \begin{align}\label{ass}
%     \frac{\kappa_1}{2}\lvert e_{x_{k,\phi}} \rvert_{\mathcal{W}_k}^2 \leq \tilde{W}_{\phi,k} \leq \kappa_2 \lvert \nabla \tilde{W}_{\phi,k}\rvert^2.
% \end{align}

% \nicksay{The right side of this cannot hold when you have collision avoidance.
% At the saddle point, you have $\lvert \nabla \tilde{W}_{\phi,k}\rvert=0$, but $\tilde{W}_{\phi,k}>0$.
% }

We follow a backstepping procedure by first treating the robot velocity $v_i$ as a virtual control input for the position subsystem \eqref{SO_a}. Specifically, we design a virtual velocity $v_{i,\phi}^*$ that drives the formation errors toward the desired formation. The actual control input $u_i$ is then designed to stabilize the coupled formation- and velocity-error dynamics through damping and passivating terms arising from the backstepping construction. 

The virtual control law is defined as
\begin{align}\label{eq:virtual_velocity}
    v_{i,\phi}^*(e_{x_{k,\phi}}) 
    = - k_1 \sum_{\substack{j\in\mathcal{N}_{i,\phi} \\ 
      \varepsilon_k = (i,j)}} a_{ij,\phi}\,\nabla_{ij} \tilde W_{\phi}(e_{x_{k, \phi}}),
    \quad i\leq N_\phi,
    %v_{\phi}^*(e_{x_\phi}) = - k_1 [E_{\phi} \otimes I_3] \nabla\bar W_{\phi}(e_{x_{k, \phi}}),
\end{align}
where $k_1>0$, $a_{ij,\phi}\in\{0,1\}$ indicates whether robots $i$ and $j$ are neighbors in mode $\phi$, that is, $a_{ij,\phi}=1$ if the edge $\varepsilon_k=(i,j) {\in \mathcal{E}_\phi}$ and $a_{ij,\phi}=0$ otherwise, and {$\nabla_{ij}\tilde W_{\phi} := 
\partial (\tilde W_{k,\phi}(e_{x_{k,\phi}})/\partial e_{x_{k,\phi}})^\top$ for $\varepsilon_k=(i,j)$ with} $\tilde W_{\phi}(e_{x_{k, \phi}})$ defined in \eqref{eq:blf_edge}. Notice that each term in \eqref{eq:virtual_velocity} depends only on the relative state associated with an edge of robot $i$. We then define the velocity-tracking error as 
\begin{equation}\label{v_tilde}
    \tilde v_{i,\phi} = v_{i,\phi} - v_{i,\phi}^*
\end{equation}
and choose the commanded acceleration as
\begin{align}\label{u}
    u_i = -&k_2 \sum_{\substack{j\in\mathcal{N}_{i,\phi} \\ 
      \varepsilon_k = (i,j)}} a_{ij,\phi}\bigl(\tilde v_{i,\phi} - \tilde v_{j,\phi}\bigr)
          - k_3 \tilde v_{i,\phi} \nonumber\\
          - &\sum_{\substack{j\in\mathcal{N}_{i,\phi} \\ 
      \varepsilon_k = (i,j)}} a_{ij,\phi}\,\nabla_{ij} \tilde W_{\phi}(e_{x_{k, \phi}})
          + \dot v_{i,\phi}^{*},
    %u = -k_2[L_\phi \otimes I_3] \tilde v - k_3 \tilde v - [E_{\phi} \otimes I_3] \nabla \bar W_{\phi} + \dot v^*,
\end{align}
where $k_2,k_3>0$.
%\nicksay{why no gain for the barrier gradient? And where does the first consensus like term comes from?}

\subsection{Open-Team Reconfiguration}
\label{sec:open_team}

We consider a \textit{formation manager} with knowledge of the current team composition, interaction topology, and robot states. The formation manager schedules robot additions and removals and assigns the edges required to preserve connectivity during each reconfiguration.

\subsubsection{\textbf{Prospective edges}}
Prior to a topology switch, the formation manager may assign a set $\mathcal E_\phi^{\mathrm p}$ of \textit{prospective edges}, i.e., edges that are intended to be added to the interaction graph. In particular, a prospective edge $\varepsilon_k=(i,j)$ may be assigned at time $t_a$ when the corresponding robots do not satisfy the nominal upper-distance constraint, i.e.,
\(
    d_k(t_a)
    =
    \lvert \delta_k(t_a)\rvert
    >
    d_{\max}.
\)
A prospective edge is treated as an ordinary interaction, except that its upper-distance constraint is temporarily relaxed from $d_{\max}$ to $d_{\max}+\rho_k$, where $\rho_k\geq0$ is an auxiliary relaxation state whose dynamics drive the relaxed bound back toward the nominal one. Until the prospective edge is physically established, the relative state required by its controller is provided either through multi-hop communication over the current interaction graph, when available, or by the formation manager.

For each prospective edge $\varepsilon_k=(i,j)\in \mathcal E_\phi^{\mathrm p}$, define the relaxed admissible set
\begin{align}\label{eq:relaxed_edge_set}
	\mathcal I_{k,\phi}^{\rho}
    := &\left\{
        e_{x_k,\phi}\in\mathbb R^3 : \right. \nonumber\\
        &\left.
    d_{\min}
        <
        \lvert
            e_{x_k,\phi}+\delta_{d_k,\phi}
        \rvert
        <
        d_{\max}+\rho_k \right\}.
\end{align}
For each mode $\phi \in \mathcal{P}$, the corresponding BLF is 
\begin{equation}\label{BLF}
    \bar W_{\phi}(e_{x,\phi},\rho_\phi)
    = \sum_{k\in\mathcal{E}_\phi}%\setminus\mathcal{E}^{\rm p}_\phi} 
        \tilde W_{\phi,k}(e_{x_k,\phi})
    + \sum_{k\in\mathcal{E}^{\rm p}_\phi} 
        \tilde W^{\rho}_{\phi,k}(e_{x_k,\phi},\rho_k),
\end{equation}
where $\tilde W^{\rho}_{\phi,k}$ is obtained from \eqref{eq:blf_edge} by replacing $d_{\max}$ with $d_{\max}+\rho_k$. %When  $\mathcal{E}^{\rm p}_\phi=\emptyset$, \eqref{BLF} reduces to the nominal potential. %, while the collision-avoidance boundary remains unchanged.

Let $t_a$ denote the instant at which $\varepsilon_k$ is assigned and let $\epsilon_\rho>0$ denote an initialization margin. We set
\begin{equation}
\label{eq:rho_initialization}
    \rho_k(t_a)
    =
    \max
    \left\{
        0,\,
        d_k(t_a)-d_{\max}+\epsilon_\rho
    \right\},
\end{equation}
and define the margin
\begin{equation}
\label{eq:relaxed_margin}
    s_k
    :=
    d_{\max}+\rho_k-d_k.
\end{equation}
Thus, $s_k>0$ if and only if the upper relaxed-distance constraint in \eqref{eq:relaxed_edge_set} is satisfied, and \eqref{eq:rho_initialization} guarantees $s_k(t_a)\geq\epsilon_\rho$.

The relaxation is progressively removed in order to recover the nominal interaction range. 

We consider the projected linear nominal contraction
\begin{equation}
\label{eq:rho_nominal_dynamics}
    \dot\rho_k
    =
    \left[-c_\rho\right]^+_{\rho_k},
    \qquad
    c_\rho>0,
\end{equation}
where, given $v\in \mathbb{R}$ and $a\in \mathbb{R}_\geq 0$, $[v]^+_a= v$ if $a>0$ and $[v]^+_a= \max\{0,v\}$ if $a=0$.
% \[
%     [v]^+_a
%     =
%     \begin{cases}
%         v, & a>0,\\
%         \max\{0,v\}, & a=0.
%     \end{cases}
% \]

%We consider the linear nominal contraction
%\begin{equation}
%\label{eq:rho_nominal_dynamics}
%    \dot\rho_k^{\rm nom}
%    =
%    -c_\rho,
%    \qquad
%    c_\rho>0,
%    \qquad
%    \rho_k>0.
%\end{equation}
However, the contraction rate of the relaxed boundary $d_{\max}+\rho_k$ must be compatible with the dynamics of the inter-robot distance $d_k$ so as to preserve $s_k>0$. To this end, we impose
\begin{equation}
\label{eq:relaxed_margin_condition}
    \dot s_k
    \geq
    -\lambda_\rho s_k,
    \qquad
    \lambda_\rho>0.
\end{equation}
Condition \eqref{eq:relaxed_margin_condition} limits the rate at which the margin $s_k$ can decrease. Specifically, for $s_k(t_a)>0$, it implies
\begin{equation}\label{eq:sol_sk}
    s_k(t)
    \geq
    s_k(t_a)e^{-\lambda_\rho(t-t_a)}
    >0,
\end{equation}
and therefore guarantees preservation of the relaxed upper-distance constraint.
Since, for $\varepsilon_k=(i,j)$,
\begin{equation}
    \dot s_k
    =
    \dot\rho_k-\dot d_k,
    \qquad
    \dot d_k
    =
    \frac{\delta_k^\top(v_i-v_j)}{d_k},
\end{equation}
condition \eqref{eq:relaxed_margin_condition} is equivalent to
\begin{equation}\label{98}
    \dot\rho_k
    \geq
    \dot d_k-\lambda_\rho s_k.
\end{equation}

Accordingly, we select the relaxation rate as the nominal linear
decay \eqref{eq:rho_nominal_dynamics} whenever it satisfies the above condition, and otherwise project
it onto the admissible boundary from \eqref{98}:
%\begin{equation}
%\label{eq:rho_dynamics}
%\dot\rho_k
%=
%\begin{cases}
%\displaystyle
%\max\left\{
%    -c_\rho,\,
%    \dot d_k-\lambda_\rho s_k
%\right\},
%& \rho_k>0,\\[2mm]
%0,
%& \rho_k=0,
%\end{cases}
%\end{equation}

\begin{equation}
\label{eq:rho_dynamics}
\dot\rho_k
=
\max\left\{
    \left[-c_\rho\right]^+_{\rho_k},\,
    \dot d_k-\lambda_\rho s_k
\right\},
\end{equation}

Hence, the nominal linear contraction is applied whenever compatible
with the relaxed constraint, while
\eqref{98} prevents the relaxed boundary from
contracting faster than allowed by the robot motion. %In particular, \eqref{eq:sol_sk}
%\eqref{eq:relaxed_margin_condition} guarantees
%\(
%    s_k(t)
%    \geq
%    s_k(t_a)e^{-\lambda_\rho(t-t_a)}
%    >0,
%\)
%and therefore preserves 
%ensures preservation of the relaxed upper-distance constraint.
Whenever the feasibility correction becomes inactive,
$\rho_k$ decreases linearly and reaches zero in finite time and it remains identically zero thereafter. The prospective edge can therefore be promoted to an established edge,
and the corresponding nominal BLF is well defined immediately after
the transition.

\begin{comment}
A prospective edge is promoted to an established edge only once the
relaxation has completely vanished and the robots lie strictly inside
the nominal interaction domain, i.e.,
\begin{equation}
\label{eq:prospective_edge_activation}
    \rho_k=0,
    \qquad
    d_{\min}+\varepsilon_d
    \leq
    d_k
    \leq
    d_{\max}-\varepsilon_d,
\end{equation}
for some $\varepsilon_d>0$. Thus, no nonzero relaxation is retained
after edge establishment, and the corresponding nominal BLF is
well defined immediately after the transition.
\end{comment}

\subsubsection{\textbf{Robot addition}}\label{par:node_addition}
{Let $\phi = \sigma(t) \in \mathcal{P}$ denote the active mode after the switching instant $t_l$, with $t \in [t_l, t_{l+1})$, and let $\hat \phi = \sigma(\hat t) \in \mathcal{P}$ denote the mode active immediately before the switching instant $t_l$, with $\hat t \in [t_{l-1}, t_l)$. Here, $t_{l-1}$, $t_l$, and $t_{l+1}$ are consecutive switching instants.}
Consider a robot $r$ requesting to join the team. The formation manager assigns its desired position and a nonempty set of prospective edges
\begin{equation}
    \mathcal E_r^{\mathrm p}
    \subseteq
    \left\{
        (r,j) : j\in\mathcal V_{\hat \phi}
    \right\},
    \qquad
    \mathcal E_r^{\mathrm p}\neq\emptyset.
\end{equation}
Since $\mathcal G_{\hat \phi}$ is connected, the establishment of any edge between $r$ and a robot in $\mathcal V_{\hat \phi}$ is sufficient to preserve connectivity after the addition of $r$. The prospective edges are established according to \eqref{eq:relaxed_edge_set}--\eqref{eq:rho_dynamics}.
Once they satisfy $\rho_k=0$, the joining event is performed as
\begin{align}
    \mathcal V_{\phi}
    =
    \mathcal V_{\hat \phi}\cup\{r\},
    \qquad
    \mathcal E_{\phi}
    =
    \mathcal E_{\hat \phi}\cup\mathcal E_r^{\mathrm p}.
\end{align}

\subsubsection{\textbf{Robot removal}}\label{par:node_removal}
Consider a robot $r\in\mathcal V_{\hat \phi}$ requesting to leave and define
\(
    \mathcal G_{\hat \phi}^{-r}
    :=
    \mathcal G_{\hat \phi}\setminus\{r\}.
\)
If $\mathcal G_{\hat \phi}^{-r}$ is connected, the robot can leave directly.
Otherwise, let
$\mathcal C_1,\ldots,\mathcal C_{n_c}$ denote the connected components
of $\mathcal G_\phi^{-r}$. The formation manager selects a set of
prospective bridging edges $\mathcal E_r^{\mathrm b}$ such that
\begin{equation}
\label{eq:departure_connectivity_condition}
    \left(
        \mathcal V_{\hat \phi}\setminus\{r\},
        \mathcal E_{\hat \phi}^{-r}\cup\mathcal E_r^{\mathrm b}
    \right)
\end{equation}
is connected, where $\mathcal E_{\hat \phi}^{-r}$ contains the edges of
$\mathcal E_{\hat \phi}$ that are not incident to $r$.

The addition of the bridging edges may render the current desired
formation incompatible with the post-departure topology. Accordingly,
the formation manager assigns a desired formation for the new mode
$\phi$, specified by desired positions
$x_{i,d,\phi}$, $i\in\mathcal V_{\hat \phi}\setminus\{r\}$.
The corresponding desired relative displacements are then defined as
\begin{equation}
    \delta_{d_{k,\phi}}
    :=
    x_{i,d,\phi}-x_{j,d,\phi},
    \qquad
    \varepsilon_k=(i,j)
    \in
    \mathcal E_{\hat \phi}^{-r}\cup\mathcal E_r^{\mathrm b}.
\end{equation}
In particular, the prospective bridging edges are established using
the desired relative displacements associated with the forthcoming
mode $\phi$.

Robot $r$ remains active while the edges in $\mathcal E_r^{\mathrm b}$ are established according to \eqref{eq:relaxed_edge_set}--\eqref{eq:rho_dynamics}. Once all required bridging edges satisfy $\rho_k=0$, the departure and edge activation are performed simultaneously,
\begin{align}
    \mathcal V_{\phi}
    &=
    \mathcal V_{\hat \phi}\setminus\{r\},\qquad \mathcal E_{\phi}
    =
    \mathcal E_{\hat \phi}^{-r}
    \cup
    \mathcal E_r^{\mathrm b}.
\end{align}
By construction, $\mathcal G_{\phi}$ is connected and the desired
relative displacements are consistent with the new interaction
topology.

\section{Stability analysis}
\subsection{Edge-based formulation}

To establish the main result, we recast the closed-loop dynamics in edge-based coordinates, where the constrained formation objective in \eqref{obj1}-\eqref{obj2} becomes the stabilization of the origin in formation error coordinates. To this end, we recall the edge-based representation induced by the incidence matrix \cite{zelazo2007agreement}. For each mode, let $E_\phi \in \mathbb{R}^{N_\phi \times M_\phi}$ denote the incidence matrix of $\mathcal{G}_\phi$ whose entries are defined as ${[E_{\phi}]_{ik}} := +1,$ if $\varepsilon_k = (j, i);$ ${[E_{\phi}]_{ik}} := -1$ if $\varepsilon_k = (i, j);$ and ${[E_{\phi}]_{ik}} := 0$, otherwise. The Laplacian and edge Laplacian matrices $L_{\phi} \in \mathbb{R}^{N_\phi \times N_\phi}$ and $L_{e_{\phi}} \in \mathbb{R}^{M_\phi \times M_\phi}$ of $\mathcal{G}_\phi$ can be expressed as
\begin{align}\label{Laplacians}
L_{\phi} = E_{\phi}E_{\phi}^\top,\quad L_{e_{\phi}} = E_{\phi}^\top E_{\phi}. 
\end{align}
The Laplacian matrix of an undirected and connected graph is symmetric positive semidefinite and has a unique zero eigenvalue, with all remaining eigenvalues being strictly positive. Then, the edge states can then be expressed in matrix form as 
\begin{equation}\label{err}
    e_{x_\phi} = [E_\phi^\top \otimes I_3] \bar x_\phi,
\end{equation}
where $e_{x_\phi}^\top = [e_{x_1}^\top\ \cdots e_{x_{M_\phi}}^\top]$ and $\bar{x}_\phi^\top = [\bar x_1^\top\ \cdots \bar x_{N_\phi}^\top]$ with $\bar x_i = x_i - x_{i,d,\phi}$. Under Assumption \ref{ass1}, since $\mathcal{G}_\phi$ contains a spanning tree, the edge states can be partitioned into those associated with the spanning tree and those associated with the remaining edges. Let $\mathcal{G}_{t_\phi}$ denote a spanning tree, and let $\mathcal{G}_{c_\phi} := \mathcal{G}_{\phi} \setminus \mathcal{G}_{t_\phi}$ denote the subgraph containing the remaining edges. Accordingly, we have $E_\phi := [E_{t_\phi} \ E_{c_\phi}]$, where $E_{t_\phi} \in \mathbb{R}^{N_\phi \times N_\phi -1}$ is the incidence matrix associated with the spanning-tree edges and $E_{c_\phi} \in \mathbb{R}^{N_\phi \times M_\phi - (N_\phi -1)}$ is the incidence matrix associated with the remaining edges. Since $E_{t_\phi}$ is the incidence matrix of a spanning tree, it has full column rank {\cite[Thm.2.1]{zelazo2007agreement}}. Hence, the corresponding edge-Laplacian matrix
$L_{e_{t,\phi}}:=E_{t_\phi}^\top E_{t_\phi}$
is positive definite. Therefore, $E_{t_\phi}^\top E_{t_\phi}$ is invertible, and there exists a matrix $R_\phi$ satisfying
\begin{align}\label{relation}
    E_\phi = E_{t_\phi} R_\phi,
\end{align}
where $R_\phi := [I_{N_\phi-1} \ T_\phi]$ and $T_\phi := (E_{t_\phi}^\top E_{t_\phi})^{-1}E_{t_\phi}^\top E_{c_\phi}$ \cite{zelazo2007agreement}. Consequently, we also have $e_\phi = [e_{x_{t_\phi}}^\top \ e_{x_{c_\phi}}^\top]^\top$, where $e_{x_{t_\phi}} \in \mathbb{R}^{(N_\phi-1) \times 3}$ are the states of the spanning-tree edges and $e_{x_{c_\phi}}\in \mathbb{R}^{(M_\phi - N_\phi+1) \times 3}$ are the states of the remaining edges. Then it follows from \eqref{relation} that
\begin{align}\label{752}
    e_{x_\phi} = [R_\phi^\top \otimes I_3] e_{x_{t_\phi}}.
\end{align}
Thus, the whole edge vector is determined by the spanning-tree edge states, and the formation control problem may be studied as a reduced-order stability problem in the coordinates $e_{x_{t_\phi}}$. We 
% define 
% \begin{align}\label{BLF}
%     \bar W_\phi (e_{x_\phi}) = \sum_{k=1}^{M_\phi}\tilde W_{\phi,k}(e_{x_{k, \phi}}),
% \end{align}
% and 
introduce the function $\bar{W}_{t,\phi}$ as ${\bar{W}_{t,\phi}(e_{x_{t_\phi}}) = \bar{W}_\phi(R^\top e_{x_{t_\phi}})}$, where $\bar W_\phi$ is defined in \eqref{BLF}.

By the chain rule and
\eqref{752},
\begin{align}
\label{641}
    \nabla \bar W_{t,\phi} = %\left(\frac{\partial \bar W_{t,\phi}} {\partial e_{x_{t_\phi}}}\right)^\top
    %=
    \left(\frac{\partial \bar W_\phi}
    {\partial e_{x_\phi}}
    \frac{\partial e_{x_\phi}}
    {\partial e_{x_{t_\phi}}}\right)^\top
    =
    (R_\phi\otimes I_3)
    \nabla_{e_{x_\phi}}\bar W_\phi.
\end{align}

 \begin{comment}
 \begin{align}
    \frac{\partial \bar W_\phi}{\partial e_{x_\phi}}
    &=
    \frac{\partial \bar W_\phi}{\partial e_{x_{t_\phi}}}
    \frac{\partial e_{x_{t_\phi}}}{\partial e_{x_\phi}}
    \nonumber\\
    &=
    \nabla \bar W_{t,\phi}^{\top} R^{\top}
    =
    \left(\nabla_{e_{x_\phi}}\bar W_\phi\right)^{\top}.
    \label{641}
\end{align}
\end{comment}
Differentiating on both sides of \eqref{err} and \eqref{v_tilde}, and using $\dot v = u$ with \eqref{u} and \eqref{641}, %and taking into account the open nature of the system, 
the closed-loop system for $t \in [t_l, t_{l+1} )$, in spanning-tree formation error coordinates, can be expressed as 
\begin{subequations}\label{152}
    \begin{align}
        \dot e_{x_{t_\phi}}(t) = -&k_1 [L_{e_t,\phi} \otimes I_3] \nabla \bar W_{t, \phi}(e_{x_{t_\phi}}(t)) \nonumber \\
        +&[E_{t,\phi}^\top \otimes I_3] \tilde v_{\phi}(t),\\
        \dot{\tilde v}_\phi(t) = -&k_2 [L_{\phi} \otimes I_3] \tilde v_\phi(t) - k_3 \tilde v_\phi(t) \nonumber\\
        - &[E_{t_\phi} \otimes I_3] \nabla\bar W_{t,\phi}(e_{x_{t_\phi}}(t)),
    \end{align}
\end{subequations}
where $L_\phi$ and $L_{e_\phi}$ are defined in \eqref{Laplacians}, and $L_{e_t,\phi} = E_{t_\phi}^\top E_{t_\phi}$ denotes the edge-Laplacian matrix associated with the spanning-tree graph. Moreover, taking into account the arrival and departure of robots described in Subsections \ref{par:node_addition}-\ref{par:node_removal}, the closed-loop system at the switching instants $t=t_l$ can be expressed as
\begin{subequations}\label{edge_switch}
    \begin{align}
        e_{x_{t_\phi}}(t_l^+) &= [\Xi^x_{\phi, \hat \phi} \otimes I_3] e_{x_{t_\phi}}(t_l^-) + \Phi^x_l \\
        \tilde v_{\phi}(t_l^+) &= [\Xi^{\tilde v}_{\phi, \hat \phi} \otimes I_3] \tilde v_{\phi}(t_l^-) + \Phi^{\tilde v}_l,
    \end{align}
\end{subequations}

The equation in \eqref{edge_switch} characterizes the edge-state transition of the switched system at each switching instant $t_l$, as described in \cite{xue2022stability}. 
%Here, $\phi = \sigma(t) \in \mathcal{P}$ denotes the active mode after switching, with $t \in [t_l, t_{l+1})$, while $\hat \phi = \sigma(\hat t) \in \mathcal{P}$ denotes the mode active immediately before the switching instant, where $\hat t \in [t_{l-1}, t_l)$. 
The matrices $\Xi^x_{\phi, \hat \phi},\Xi^{\tilde v}_{\phi, \hat \phi} \in \mathbb{B}^{M_{\phi} \times M_{\hat \phi}}$ have entries $\{ 0, 1\}$ and determine how the dimension of the error states vector transition between two modes. In particular, $e_{x_t,\hat \phi}(t_l^-)$ represents the error states vector just before the switching instant, while $e_{x_t,\phi}(t_l^+)$ represents the error states vector just after the switching instant. $\Phi^x_l,\Phi^{\tilde v}_l \in \mathbb{R}^{3M_{\phi}}$ are real-valued and bounded vectors collecting the discontinuities introduced by the switching. Their only non-zero components are  for each newly activated edge and for each robot whose neighbour set changes. Both are bounded: an edge is activated only when the corresponding pair is bounded away from both distance constraints, so that $\lvert\Phi^x_l\rvert\leq\bar\Phi_x$, and $\lvert v^{*}_{i,\phi}\rvert<\bar v_i$ gives $\lvert\Phi^{\tilde v}_l\rvert\leq\bar\Phi_{\tilde v},\ \bar \Phi_{\tilde v} >0$.
% that captures instantaneous changes resulting from events such as the addition and removal of a robot or the formation or removal of an edge.  

\subsection{Stability under topology switching}

First, we recall the concept of transition-dependent average dwell time for the switching signal $\sigma(t)$, which ensures that the switching signal meets the required conditions for stability. We first impose the following assumption on the set of switching modes, that implies an upper bound on $N_\phi$.
\begin{ass}\label{ass2}
    The total number of possible switching modes is finite, that is, $\mbox{card}(\mathcal{P}) < \infty$.
\end{ass}
\begin{definition}\label{def2} (\cite[Definition 2]{xue2022stability}) On a given interval $[t_0, t_f)$, with $t_f > t_0 \geq 0$, consider any two consecutive modes $\hat{\phi}, \phi \in \mathcal{P}$, where $\hat{\phi}$ precedes $\phi$. Let $N_{\hat{\phi},\phi}(t_0, t_f)$ denote the total number of switchings from mode $\hat{\phi}$ to mode $\phi$, and let $T_{\phi}$ denote the total active duration of mode $\phi$. The constant $\tau_{\hat{\phi},\phi} > 0$ satisfying $N_{\hat{\phi},\phi}(t_0, t_f) \leq \hat{N}_{\hat{\phi},\phi} + \frac{T_{\phi}(t_0, t_f)}{\tau_{\hat{\phi},\phi}},$ for any given scalar $\hat{N}_{\hat{\phi},\phi} \geq 0$, is called the \textit{transition-dependent average dwell time} of the switching signal $\sigma(t)$.
\end{definition}
Definition \ref{def2}, we are now ready to present our main stability result.
\begin{proposition}\label{prop1}
Consider the OMRS \eqref{SO}, under Assumptions \ref{ass1}-\ref{ass2}, in closed loop with the switching control law \eqref{u}, where $\bar W_{t,\phi}$ is defined in \eqref{BLF}. Let $\phi, \hat \phi \in \mathcal{P}$ be any two consecutive modes, where $\hat \phi$ precedes $\phi$. Define $\Omega_{\phi,\hat{\phi}}=2$, $\gamma_{\phi}=\min \{ \frac{k_1\lambda_{\min}(L_{e_{t,\phi}})}{\kappa_2},2k_3 \}$, with $\kappa_2>0$ a positive constant. Let $\bar c = \min_{\phi\in\mathcal{P}}\bar W_{t,\phi}(e_{t,\phi}^*)>0$, where $e^*_{t,\phi}$ denotes the saddle point of $\bar W_{t,\phi}$, and for $c\in(0,\bar c)$ define $\mathcal{S}_{\phi} := \{ e_{x_t,\phi} \in \mathcal{I}_\phi : \bar W_{t,\phi}(e_{x_{t_\phi}}) \leq c\}$. If the switching signal $\sigma$ admits a transition-dependent average dwell time satisfying
	\begin{equation}\label{cond}
		\tau_{\phi,\hat{\phi}} \geq \frac{\ln(\Omega_{\phi,\hat{\phi}})}{\gamma_{\phi}},
\end{equation} 
and if $c_0>0$ satisfies
\begin{equation}\label{eq:smallness}
    \Omega_{\phi,\hat\phi}\, c_0 + \Theta  \leq c,
\end{equation}
where %$\Omega_{\phi,\hat\phi}$ represents the factor by which the energy of the formation may grow and 
$\Theta$ represents the energy added to the system at a switching instant defined further in the proof, then
%For every mode $\phi \in \mathcal{P}$ the constraints set $\mathcal{I}_{\phi}$ defined in \eqref{294} is forward invariant, hence, collisions are avoided and all the initial and added edges are maintained. 
the origin of the closed-loop system \eqref{152}-\eqref{edge_switch} is uniformly practically stable for all initial conditions such that $(e_{x_{t_\phi}}(0), \tilde v_\phi (0)) \in \mathcal{S}_{\phi} \times \mathbb{R}^{3N_\phi}$. %, with a residual enlarged by $\varsigma/\gamma_\phi$ during the intervals on which prospective edges are pending, where $\varsigma \geq 0$. 
Moreover, for every mode $\phi \in \mathcal{P}$, the constraints set $\mathcal{I}_\phi$ defined in \eqref{294} is forward invariant.
%Moreover, collisions are avoided and all the initial and added edges are maintained for all $t \geq t_0$. %Moreover, under Assumption \ref{ass3}, the origin of the closed-loop system \eqref{152}-\eqref{edge_switch} is asymptotically stable for almost all initial conditions.
\end{proposition}

\begin{proof}
For each mode $\phi \in \mathcal{P}$, we consider the Lyapunov function candidate
    \begin{align}\label{V}
        V_\phi = \bar W_{t,\phi} + \frac{1}{2}\lvert \tilde v_{\phi} \rvert^2,
    \end{align}
where $\bar W_{t,\phi}$ satisfies
\begin{align}\label{ass}
 \bar{W}_{t,\phi} \leq \kappa_2 \lvert \nabla \bar{W}_{t,\phi}\rvert^2,
\end{align}
where $\kappa_2>0$. First, we consider the case $\mathcal{E}^{\rm p}_\phi = \emptyset$. For all $\tau \in [t_l, t_{l+1})$, the derivative of \eqref{V} along the trajectories of \eqref{152} yields
\begin{align*}
    \dot V %\nabla& \bar W_{t,\phi}^\top \dot{e}_{{x_{t, \phi}}} + \tilde v^\top \left( -k_2 L_{\phi}  \tilde v - k_3 \tilde v - E_{t_\phi} \nabla\bar W_{t,\phi} \right)\\
    = -&k_1 \nabla \bar W_{t,\phi}^\top [L_{e_{t,\phi}}\otimes I_3]\nabla \bar W_{t,\phi} + \nabla \bar W_{t,\phi}^\top [E^\top_{t,\phi} \otimes I_3] \tilde v_{\phi}\\
    -&k_3  \tilde v_\phi^\top \tilde v_\phi -\tilde v_\phi^\top [E_{t,\phi}\otimes I_3]\nabla \bar W_{t,\phi} - k_2 \tilde v_\phi^\top [L_\phi\otimes I_3] \tilde v_\phi \\
    \leq -&k_1 \lambda_{\min}(L_{e_{t,\phi}})\lvert \nabla \bar W_{t,\phi} \rvert^2 - k_3 \lvert \tilde v_\phi \rvert^2
\end{align*}
due to that since $L_\phi \geq 0$ and $L_{e_{t,\phi}}>0$, we have $\lambda_{\min}(L_\phi)=0$ and $0 < \lambda_{\min}(L_{e_{t,\phi}})$.
Consequently, given $\gamma_\phi = \min \{ \frac{k_1\lambda_{\min}(L_{e_{t,\phi}})}{\kappa_2},2k_3 \}$ and from \eqref{ass}, we obtain
\begin{align}\label{cond1}
	\dot V_{\phi}(\bar e_{x_t,\phi}(\tau),\tilde v_\phi(\tau)) \leq -\gamma_{\phi} V_{\phi}(\bar e_{x_t,\phi}(\tau),\tilde v_\phi(\tau)).
\end{align}
Now, let $\phi, \hat \phi \in \mathcal{P}$ be two consecutive modes and  $V_{\phi}(\bar e_{x_t,\phi}(t_l^+),\tilde v_\phi(t_l^+)) := \bar W_{t,\phi}(t_l^+)+ \frac{1}{2}\lvert \tilde v_\phi(t_l^+)\rvert^2$ and
    $V_{\phi}(\bar e_{x_t,\phi}(t_l^-),\tilde v_\phi(t_l^-)) := \bar W_{t,\phi}(t_l^-)+ \frac{1}{2}\lvert \tilde v_\phi(t_l^-)\rvert^2.$
% \begin{align*}
%     V_{\phi}(\bar e_{x_t,\phi}(t_l^+),\tilde v_\phi(t_l^+)) &:= \bar W_{t,\phi}(t_l^+)+ \frac{1}{2}\lvert \tilde v_\phi(t_l^+)\rvert^2\\
%     V_{\phi}(\bar e_{x_t,\phi}(t_l^-),\tilde v_\phi(t_l^-)) &:= \bar W_{t,\phi}(t_l^-)+ \frac{1}{2}\lvert \tilde v_\phi(t_l^-)\rvert^2.
% \end{align*}
% $$V_{\phi}(\bar e_{x,\phi}(t_l^+),\tilde v_\phi(t_l^+)):= \bar W_{t,\phi}(t_l^+)+ \frac{1}{2}\lvert \tilde v_\phi(t_l^+)\rvert^2$$ $V_{\phi}(\bar e_{x,\phi}(t_l^-),\tilde v_\phi(t_l^-)):= \bar W_{t,\phi}(t_l^-)+ \frac{1}{2}\lvert \tilde v_\phi(t_l^-)\rvert^2.$ 
It follows from \eqref{edge_switch} that for any $t_l$, 
%\begin{subequations}\label{inequalities}
\begin{align}
    %\lvert e_{x_t,\phi}(t_l^+) \rvert^2 &\leq 2\vert \Xi^x_{\phi, \hat \phi} \rvert^2 \lvert e_{x_t,\hat\phi}(t_l^-) \rvert^2 + 2\lvert \Phi^x_l\rvert^2 \label{inequalities_a}\\
    \lvert \tilde v_{\phi}(t_l^+) \rvert^2 &\leq 2\vert \Xi^{\tilde v}_{\phi, \hat \phi} \rvert^2 \lvert \tilde v_{\hat \phi}(t_l^-) \rvert^2 + 2\lvert \Phi^{\tilde v}_l\rvert^2. \label{inequalities_b}
\end{align}
%\end{subequations}
% \begin{align*}
%     \lvert e_{x_t,\phi}(t_l^+) \rvert^2 &\leq 2\vert \Xi^x_{\phi, \hat \phi} \rvert^2 \lvert e_{x_t,\hat\phi}(t_l^-) \rvert^2 + 2\lvert \Phi^x_l\rvert^2\\
%     \lvert \tilde v_{\phi}(t_l^+) \rvert^2 &\leq 2\vert \Xi^{\tilde v}_{\phi, \hat \phi} \rvert^2 \lvert \tilde v_{\hat \phi}(t_l^-) \rvert^2 + 2\lvert \Phi^{\tilde v}_l\rvert^2.
% \end{align*}
% $\lvert e_{x,{\phi}}(t_l^+) \rvert^2 \leq 2\vert \Xi^x_{\phi, \hat \phi} \rvert^2 \lvert e_{x,{\hat \phi}}(t_l^-) \rvert^2 + 2\lvert \Phi^x_l\rvert^2$ and $\lvert \tilde v_{\phi}(t_l^+) \rvert^2 \leq 2\vert \Xi^{\tilde v}_{\phi, \hat \phi} \rvert^2 \lvert \tilde v_{\hat \phi}(t_l^-) \rvert^2 + 2\lvert \Phi^{\tilde v}_l\rvert^2$.
We note that $\vert \Xi^x_{\phi, \hat \phi} \rvert \equiv 1$ and $\vert \Xi^{\tilde v}_{\phi, \hat \phi} \rvert \equiv 1$ because they are submatrices of an identity matrix, and the singular values of an identity matrix are all 1, meaning that the largest singular value is also 1. At each switching instant $t=t_l$, let $\varrho := M_\phi - M_{\hat\phi}$ denote the number of newly activated edges.  
%$\tilde W_{\phi,k}$ is continuous on the compact set defined by \eqref{eq:prospective_edge_activation}, and 
Under Assumption \ref{ass2}$, \mathcal{P}$ is finite and, there exists $\bar W_\epsilon>0$ such that
\begin{equation}\label{eq:Wbar_eps}
    \tilde W_{\phi,k}\bigl(e_{x_{k,\phi}}(t_l^+)\bigr) \;\leq\; \bar W_\epsilon, \quad M_{\hat\phi} < k \leq M_\phi ,
\end{equation}
at every switching instant $t_l$ and for every $\phi \in \mathcal{P}$. Since the states of the edges already present are unaffected by the transition, it follows from \eqref{edge_switch} and \eqref{eq:Wbar_eps} that
\begin{align*}
    &V_{\phi}(\bar e_{x_t,\phi}(t_l^+),\tilde v_\phi(t_l^+)) = \bar W_{t, \hat \phi}(t_l^-) + \frac{1}{2} \lvert \tilde v_\phi (t_l^+) \rvert^2\\
    &+ \sum_{k=M_{\hat \phi}+1}^{M_\phi} W_{t, \phi, k}(e_{x_t,k,\phi}(t_l^+)) \\
    &\leq \bar W_{t, \hat \phi}(t_l^-) + \varrho \bar W_\epsilon + \frac{1}{2} \lvert \tilde v_\phi (t_l^+) \rvert^2%\\
     %&\leq \bar W_{t, \hat \phi}(t_l^-) + \varrho \bar W_\epsilon + \frac{1}{2}\left( 2\vert \Xi^{\tilde v}_{\phi, \hat \phi} \rvert^2 \lvert \tilde v_{\hat \phi}(t_l^-) \rvert^2 + 2\lvert \Phi^{\tilde v}_l\rvert^2 \right)
     % &\leq \left( 1 + \frac{4 \kappa_2'}{\kappa_1} \vert \Xi^x_{\phi, \hat \phi} \rvert^2\right) \bar W_{t, \hat \phi}(t_l^-) + \vert \Xi^{\tilde v}_{\phi, \hat \phi} \rvert^2 \lvert \tilde v_\phi(t_l^-)\rvert^2\\
     % &+ 2\kappa'_2 \lvert \Phi^x_l\rvert^2 + \lvert \Phi^{\tilde v}_l \rvert^2,
\end{align*}
and therefore, using $\bar W_{t,\hat\phi}(t_l^-) \leq 2\bar W_{t,\hat\phi}(t_l^-)$ and \eqref{inequalities_b}, we obtain
	\begin{align}\label{cond2}
		V_{\phi}(e_{x_t,\phi}(t_l^+), \tilde v_\phi (t_l^+)) \leq \Omega_{\phi , \hat \phi} V_{\hat \phi}(e_{x_t,\phi}(t_l^-), \tilde v_\phi (t_l^-)) + \Theta,
	\end{align}
where $\Omega_{\phi , \hat \phi} = 2 $ and $\Theta = \bar \varrho \bar W_\epsilon + \bar \Phi_{\tilde v}$ with $\bar\varrho := \max_{l} \varrho$ is the largest number of edges activated at a single switching instant. Then, from \eqref{eq:smallness}, \eqref{cond1}, \eqref{cond2} and invoking \cite[Theorem 1]{xue2022stability}, it follows that the origin of system \eqref{152}-\eqref{edge_switch} is uniformly practically stable for all initial conditions in $\mathcal S_\phi \times \mathbb{R}^{3N_\phi}$ if the switching $\sigma$ admits an average dwell-time that satisfies \eqref{cond}.

When $\mathcal{E}^{\rm p}_\phi \neq \emptyset$, the derivative of \eqref{V} carries the 
additional term $\sum_{k\in\mathcal{E}^{\rm p}_\phi}
\frac{\partial W^{\rho}_{\phi,k}}{\partial \rho_k}\,\dot\rho_k.$ Enlarging $\rho_k$ relaxes the upper bound in \eqref{eq:relaxed_edge_set} and hence decreases the barrier, so $\partial W^{\rho}_{\phi,k}/\partial\rho_k \leq 0$. From \eqref{eq:rho_dynamics}, if \eqref{eq:rho_nominal_dynamics} is active then $\dot \rho_k < 0$ and if \eqref{98} is active then $\dot \rho_k > 0$. In the latter case it remains bounded from \eqref{eq:relaxed_margin_condition}. Consequently there exists $\varsigma\geq0$ such that $\sum_{k\in\mathcal{E}^{\rm p}_\phi}\frac{\partial W^{\rho}_{\phi,k}}{\partial \rho_k}\,\dot\rho_k \;\leq\; \varsigma,$ and \eqref{cond1} is replaced by $\dot V_{\phi} \leq -\gamma_{\phi} V_{\phi} + \varsigma.$ Therefore, the argument of the previous step holds with $\Theta$ replaced by 
$\Theta_\varsigma := \Theta + \Omega_{\phi,\hat\phi}\varsigma/\gamma$, where 
$\gamma := \min_{\phi\in\mathcal{P}}\gamma_\phi$. Moreover $\varsigma = 0$ once every prospective edge has been activated.

Next, we show that the set $\mathcal{I}_{\phi}$ is forward invariant and corresponds to the domain of attraction for the closed loop system \eqref{152}-\eqref{edge_switch}. We proceed by contradiction. From \eqref{cond1}, and the boundedness of $\rho_k$ implied by \eqref{eq:rho_dynamics}, when $\mathcal{E}^{\rm p}_\phi \neq \emptyset$, for any mode $\phi = \sigma(t)$, $V_{\sigma(t)}$ remains bounded on each interval $[t_l,t_{l+1})$, and from \eqref{cond2} together with \eqref{eq:smallness}, it satisfies $V_{\sigma(t)}(t_l^{+})\leq c$ at every switching instant. Hence $V_{\sigma(t)}(t)\leq c$ for all $t\geq t_0$. Assume now that there exists $T>t_0$ such that $e_{x_t,\sigma(t)}(t)\in\mathcal{I}_{\sigma(t)}$ for all $t\in[t_0,T)$ but $e_{x_t,\sigma(T)}(T)\notin\mathcal{I}_{\sigma(T)}$. Then, $\lvert e_{x_{k},\sigma(t)}+\delta_{d_{k},\sigma(t)}\rvert$ tends to either $d_{\min}$ or $d_{\max}$ as $t\to T$, for at least one $k\leq M_{\sigma(t)}$. By the definition of the BLF \eqref{eq:blf_edge}, this implies $V_{\sigma(t)}(t)\to\infty$ as $t\to T$, contradicting $V_{\sigma(t)}(t)\leq c$ for all $t\geq t_0$. Forward invariance of $\mathcal{I}_{\sigma(t)}$ follows. When $\mathcal{E}^{\rm p}_\phi \neq \emptyset$, the same argument applies with $\mathcal{I}_{\sigma(t)}$ replaced by $\mathcal{I}^{\rho}_{\sigma(t)}$ in \eqref{eq:relaxed_edge_set}. %, so that inter-robot collisions are avoided and all initial and newly activated edges are maintained. 
By construction $c<\bar c$, so $\mathcal{S}_{\phi}$ excludes the saddle point $e^{*}_{t,\phi}$ and the origin is the only critical point of $\bar W_{t,\phi}$ in $\mathcal{S}_{\phi}$. Moreover $V_\phi$ is positive definite on $\mathcal{S}_{\phi}\times\mathbb{R}^{3N_\phi}$ and, by the forward invariance established above, the trajectories remain in $\mathcal{S}_{\sigma(t)}\times\mathbb{R}^{3N_\phi}$ for all $t\geq t_0$, where \eqref{cond1} and \eqref{cond2} hold. %Invoking \cite[Theorem~1]{xue2022stability}, the origin is uniformly practically stable for all initial conditions in $\mathcal{S}_{\phi_0}\times\mathbb{R}^{3N_{\phi_0}}$.

\end{proof}
\section{Simulation Results}
\label{sec:simulation}

The proposed framework is validated in a ROS~2--Gazebo simulation with seven quadrotors. Gazebo simulates the full six-degree-of-freedom rigid-body dynamics of each vehicle together with the individual rotor dynamics. The distributed controller developed in the previous sections provides the desired translational acceleration, which is realized by a low-level geometric controller \cite{lee2010geometric} that computes the commanded total thrust and attitude, followed by rotor allocation. A charging station is included in the environment, where robots that are not currently part of the team remain parked and from which newly joining robots are deployed (Fig.~\ref{fig:gazebo_setup}).
\begin{figure}
    \centering
    \includegraphics[width=0.6\linewidth]{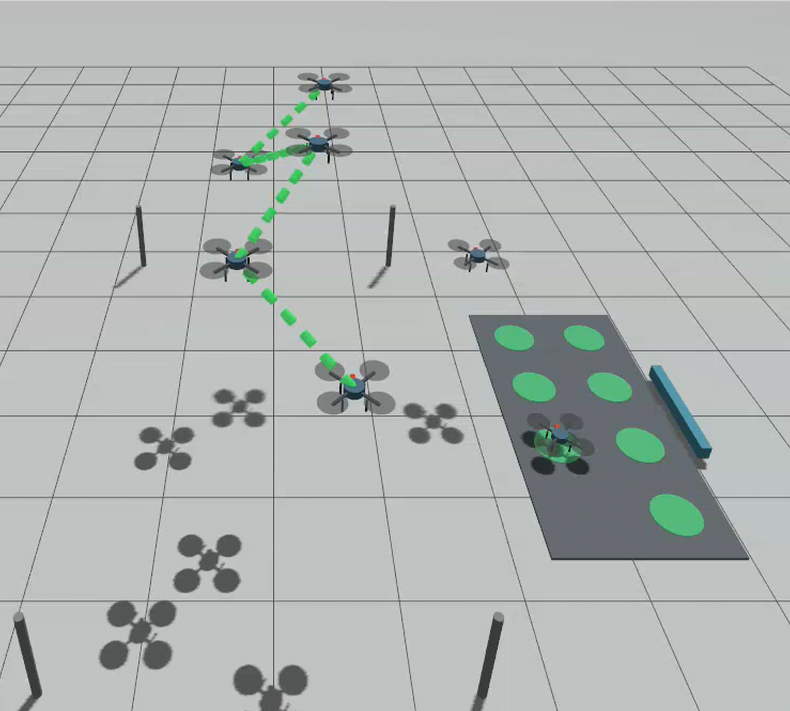}
    \caption{Simulation setup, showing the charging station (on the right) with a drone parked and another in the process of joining the active team.}
    \label{fig:gazebo_setup}
\end{figure}

For the numerical validation, the edge potentials are instantiated using
the weighted recentered barrier function of~\cite{feller2015weight}.
For each $k\leq M_\phi$, let
\(
    \delta_{k,\phi}
    =
    e_{x_k,\phi}+\delta_{d_k,\phi}
\)
denote the corresponding relative displacement. Let $\kappa_{1,k}
=
\frac{1}{2}
\frac{d_{\min}^2}
{
    \lvert\delta_{d_k,\phi}\rvert^2
    \left(
        \lvert\delta_{d_k,\phi}\rvert^2-d_{\min}^2
    \right)
}$ and $\kappa_{2,k}=
\frac{1}{2}
\frac{1}
{
    d_{\max}^2-\lvert\delta_{d_k,\phi}\rvert^2
}.$ We use
\begin{align}\label{eq:simulation_blf}
&\widetilde W_{k,\phi}(e_{x_k,\phi})
={}
\frac{1}{2}\lvert e_{x_k,\phi}\rvert^2
\nonumber\\
&+
\kappa_{1,k}
\left[
    \ln\left(
        \frac{d_{\max}^2}
        {d_{\max}^2-\lvert\delta_{k,\phi}\rvert^2}
    \right)
    -
    \ln\left(
        \frac{d_{\max}^2}
        {d_{\max}^2-\lvert\delta_{d_k,\phi}\rvert^2}
    \right)
\right]
\nonumber\\
&+
\kappa_{2,k}
\left[
    \ln\left(
        \frac{\lvert\delta_{k,\phi}\rvert^2}
        {\lvert\delta_{k,\phi}\rvert^2-d_{\min}^2}
    \right)
    -
    \ln\left(
        \frac{\lvert\delta_{d_k,\phi}\rvert^2}
        {\lvert\delta_{d_k,\phi}\rvert^2-d_{\min}^2}
    \right)
\right].
\end{align}
% where
% \begin{equation}
% \begin{aligned}
% \label{eq:simulation_blf_weights}
% \kappa_{1,k}
% &=
% \frac{1}{2}
% \frac{d_{\min}^2}
% {
%     \lvert\delta_{d_k,\phi}\rvert^2
%     \left(
%         \lvert\delta_{d_k,\phi}\rvert^2-d_{\min}^2
%     \right)
% },
% \\
% \kappa_{2,k}
% &=
% \frac{1}{2}
% \frac{1}
% {
%     d_{\max}^2-\lvert\delta_{d_k,\phi}\rvert^2
% }.
% \end{aligned}
% \end{equation}
% The weights are selected such that the barrier contribution has zero gradient at the desired relative displacement, while the potential diverges as either distance boundary is approached. 
For prospective edges, the same construction is employed with the relaxed upper bound $d_{\max}+\rho_k$ in place of $d_{\max}$.

The simulation parameters are
$d_{\min}=0.45~\mathrm{m}$,
$d_{\max}=1.65~\mathrm{m}$,
$(k_1,k_2,k_3)=(0.30,0.30,1.25)$,
$c_\rho=0.10~\mathrm{m/s}$,
$\lambda_\rho=2.0~\mathrm{s}^{-1}$, and
$\varepsilon_\rho=0.22~\mathrm{m}$.
The experiment considers seven robots, five of which initially form the
chain
\(
(1,2),(2,3),(3,4),(4,5)
\),
while robots~6 and~7 are initially at the charging station and request
admission at $t=8~\mathrm{s}$ and $t=25~\mathrm{s}$, respectively.
They sequentially take off, approach the team, establish their prospective
interactions, and are admitted to the OMRS according to
Section~\ref{par:node_addition}. Subsequently, robots~3, 2, and~4 request
departure at $t=51~\mathrm{s}$, $73~\mathrm{s}$, and $95~\mathrm{s}$.
Whenever a departure would disconnect the interaction graph, the formation
manager first establishes the required prospective bridge and removes the
robot only after the bridge has become an established interaction, as
described in Section~\ref{par:node_removal}. Departing robots then
independently return to the charging station, resulting in repeated changes
in both the number of active robots and the interaction topology.

Figure~\ref{fig:plots} summarizes the evolution of the system over the complete mission. 
The center panels report the formation-control performance. Starting from a large initial error, the desired formation is rapidly recovered. Reconfiguration events introduce temporary transients, as expected from the changes in the desired formation and interaction graph, after which the errors converge again. 
The composite Lyapunov function exhibits the same behavior, decreasing between consecutive reconfiguration events.
\begin{figure*}[t]
    \centering
    \includegraphics[width=\textwidth]{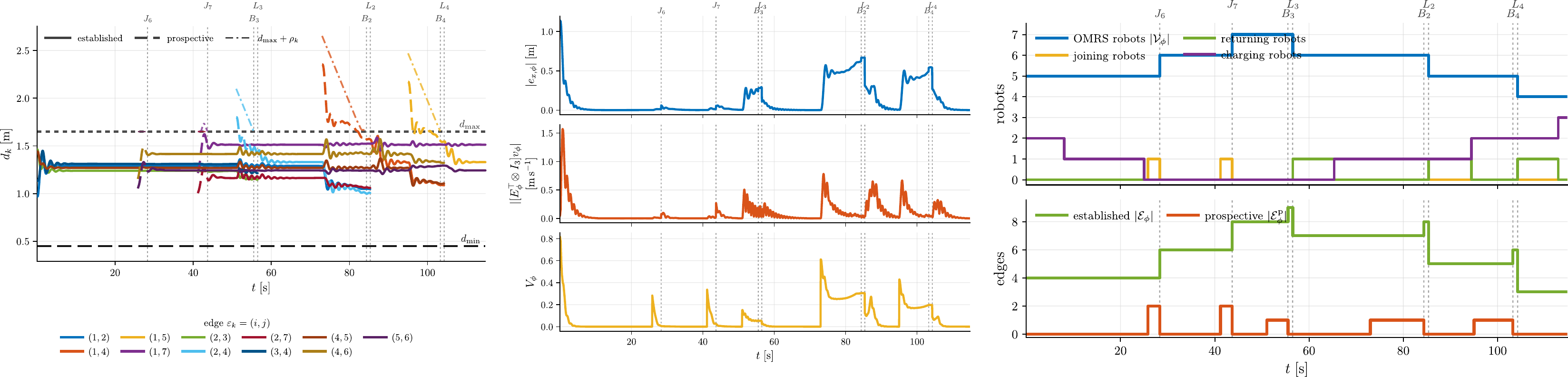}
    \caption{Evolution of the OMRS during the Gazebo simulation.
    Left: inter-robot distances associated with established and prospective edges, together with the relaxed upper bounds $d_{\max}+\rho_k$ and the nominal distance limits $d_{\min}$ and $d_{\max}$.
    Center: formation-position error $\lvert e_{x,\phi}\rvert$, velocity disagreement $\lvert [E_\phi^\top\otimes I_3]v_\phi\rvert$, and composite Lyapunov function $V_\phi$.
    Right: number of robots in the OMRS, joining, returning, and at the charging station (top), and number of established and prospective edges (bottom).
    Vertical dashed lines indicate robot admissions $J_i$, bridge activations $B_i$, and robot departures $L_i$.}
    \label{fig:plots}
\end{figure*}
The left panel of Fig.~\ref{fig:plots} shows the evolution of the inter-robot constraints. Throughout the simulation, the minimum distance between any two physical robots remains above the collision-avoidance bound, while every established interaction remains within the nominal connectivity range.
Prospective interactions are allowed to start outside the nominal range; in the considered mission they are initialized at distances of up to approximately $2.35~\mathrm{m}$. The corresponding relaxed bounds $d_{\max}+\rho_k$ preserve feasibility while the prospective robots approach one another, with $\rho_k$ converging to zero before the associated interactions are promoted to established edges. In particular, three bridge interactions are successfully created before the three planned robot departures. Finally, the right panels of Fig.~\ref{fig:plots} highlight the open nature of the system, showing the changes in the number of robots and interaction edges throughout the mission. Two robots are sequentially admitted to the OMRS, while three robots subsequently leave and return to the charging station, with connectivity preserved through the make-before-break bridge construction.

Overall, the results show that the proposed controller and formation manager preserve collision avoidance and connectivity while repeatedly adding and removing robots from a multi-UAV system. The complete evolution of the experiment, including robot admission, bridge formation, departures, and return-to-charge maneuvers, is shown in the accompanying video.

\section{CONCLUSIONS}
We presented a BLF-based distributed control framework for formation control of OMRS with robots joining and leaving the team over time. The proposed formation manager coordinates these reconfigurations by defining new edges to be formed for connectivity-critical departures and temporarily relaxing their connectivity constraints through auxiliary dynamics that tend to recover the nominal interaction range in finite time, while the controller enforces collision avoidance and connectivity constraints. The resulting switched closed-loop system was shown to be uniformly practically stable under a transition-dependent average dwell-time condition, and the approach was validated in realistic multi-UAV Gazebo simulations.


\bibliographystyle{IEEEtran}
\bibliography{references_pelin}

\end{document}